\documentclass{amsart}[11pt]
\usepackage[left=40mm,right=40mm,top=2cm,bottom=2cm]{geometry}
\usepackage{url}
\usepackage[colorlinks=true,citecolor=blue,linkcolor=blue,urlcolor=blue,bookmarks,bookmarksopen,bookmarksdepth=2,backref=page,breaklinks]{hyperref}
\usepackage{float}
\usepackage{enumerate}  
\usepackage{xcolor}
\usepackage{colortbl}
\usepackage{multirow}
\usepackage{multicol}
\usepackage{wrapfig}
\long\def\gloops#1{}
\usepackage{environ}
\NewEnviron{noproof}{}

\newtheorem{example}{Example}
\newtheorem{scolie}{Scolie}

\newtheorem{lemma}{Lemma}

\def\f2{{\mathbb F}_2}
\def\fd(#1){{\mathbb F}^{#1}_2}
\def\ffd(#1){{\mathbb F}_{2^{#1}}}
\def\Comp(#1){{\rm comp}(#1)}
\def\comp(#1){{\langle #1 \rangle}}
\def\abs(#1){\vert #1\vert}

\def\tfr(#1,#2){ { {\widehat{#1}(#2)} } }
\def\TFR(#1,#2,#3){ { {\widehat{#1}(#2,#3)} } }
\def\spec(#1){  {\rm spec}{\,(#1)}}  
\def\moment(#1,#2){ { \kappa_{#1}(#2)} }
\def\triv(#1,#2){ { AROUND{\rm triv}_{#2}(#1)} }
\def\inv{  \mathfrak J } 
\def\invwalsh{  \mathfrak w }
\def\invec{{\mathfrak J}'}
\def\QQ{ \mathcal{Q}}
\def\boole(#1){ B(#1) }
\def\nl(#1){  {\rm nl}{\,(#1)}} 
\def\l(#1){  {\rm l}{\,(#1)}} 
\def\CCZ{\textsc{ccz }}
\def\CCZ-{\textsc{ccz}-}
\def\APN{\textsc{apn} }
\def\EA{\textsc{ea} }
\def\MNBC{\textsc{mnbc} }
\def\EA-{\textsc{ea}-}
\def\ext(#1){{\rm Ext}(#1)}
\def\bleu#1{\textcolor{black}{#1}}
\def\rouge#1{\textcolor{black}{#1}}
\title{Spectral Moment of Order Four and the Uniqueness of the CCZ class of Dublin APN Permutation}
\author{Valérie Gillot}

\author{Philippe Langevin}
\author{Abdoulaye Lo}

\begin{document}

\maketitle

\section{Introduction}
At Finite Field conference FQ9 in Dublin,  K. A. Browning, J. F. Dillon, M. T. McQuistan, and A. J. Wolfe 
offered to us very good news : the discovery of an \APN permutation in
dimension six \cite{DILLON}, we will refer to  this as "Dublin permutation", or
\CCZ-class of Dublin. Since this announcement of their "update", 
numerous attempts have been made to find a new \APN permutation in even dimension, 
\bleu{unfortunately, no new such permutations have been found!}
\textcolor{black}{Thanks to the article 
\cite{EDEL},
\rouge{we now have a list of 14 \CCZ-classes of 6-bit \APN functions}
with representatives of degree less than or equal to 3. The two numerical searches in dimension 6, the classification 
of cubic \APN functions \cite{LANGEVIN} and the search 
\rouge{for \APN permutations} \cite{LEANDER} suggest that there is no other \CCZ-class} \rouge{in dimension 6}. However, 
the lack of theoretical results leaves open the possibility of unknown sporadic classes,
hidden within the complexity of the combinatorial of the problem. In order to eventually 
uncover a novelty in dimension 6, one must search among functions of degree greater 
than or equal to 4, probably equal. In this talk, we address the question of the existence 
of an \APN function of degree 4 having a special structure based on
observations of the decomposition of the 14 known \CCZ-classes \cite{CALDERINI}.
\bleu{More precisely, 12 of the 14 known \CCZ-classes contain at least one 
class composed of vectorial functions such that the set of fourth-order spectral moments 
of the components has exactly two distinct values.}
We present a procedure 
to classify 6 bits \APN quartics sharing this regularity. To achieve this, we introduce a new 
algorithm to test the existence of an APN extension of a given $(m,m-2)$-function. Our talk also provides specific results on APN-functions based on the 
classification of 6-bits Boolean functions. The  technical details are developped 
in the following sections.  

\section{Boolean and vectorial function}
Let $\f2$ be the finite field of order $2$. Let $m$ be a positive integer. 
We denote $\boole(m)$ the set of Boolean functions $f \colon \fd(m) \rightarrow\f2$. Every Boolean 
function has a unique algebraic reduced representation:
\begin{equation}\label{ANF}
f(x_1, x_2, \ldots, x_m ) = f(x) = \sum_{S\subseteq \{1,2,\ldots, m\}} a_S X_S,
\quad a_S\in\f2, \ { X_S = \prod_{s\in S} x_s}.
\end{equation}

The \textsl{degree} of $f$ is the maximal cardinality of $S$ with  $a_S=1$ in
the algebraic form. In this paper, we conventionally fix the degree
of the null function to zero. To classify Boolean functions, one introduces two definitions of equivalency, for $f,g \in \boole(m)$,  $f$ and $g$ are \textsl{affine equivalent} (equivalent) if there exist an affine permutation $A$ of $\fd(m)$ such that \rouge{$(f \circ A)(x) = g(x)$ or $g(x)+1$} ; $f$ and $g$ are \textsl{extended affine equivalent} (\EA-equivalent) if there exist an affine permutation $A$ of $\fd(m)$ and an affine Boolean function $\ell$  such that $g(x)=( f \circ A)(x)+\ell(x)$. The \textsl{Walsh coefficient} of $f\in \boole(m)$ at $a \in \fd(m)$ is 
$$\tfr (f,a) =\sum_{x \in \fd(m)} (-1)^{f(x)+a.x},$$ the multiset of 
Walsh coefficients is called the \textsl{Walsh spectrum}. Let $q=2^m$, the Walsh 
coefficients satisfy Parseval's identity :
\begin{equation}
\label{PARSEVAL}
 \sum_{a\in \fd(m)} {\tfr(f,a)}^2 = q^2
\end{equation}

The \textsl{spectral moment of order} $r$ is the integer :
\begin{equation}\label{MOMENT}
 \moment(r,f)= 
\frac 1{q^2}\sum_{ a \in \fd(m)  } \tfr( f , a )^r.
\end{equation}
It is an \EA-invariant \textcolor{black}{when $r$ is even} and we normalize the 4th-order spectral moment~:
\begin{equation}
\label{KAPPA}
 \kappa(f)= \frac 1q \moment(4, f)=\frac{1}{q^3}\sum_{a\in \fd(m)} {\tfr(f,a)}^4 
\end{equation}
The multiset 
$\invwalsh(f) := \{\!\!\{ \abs( {\tfr(f,a)} ) \mid a \in \fd(m) \}\!\!\}$
of absolute value of Walsh coefficients of $f$, is an \rouge{\EA-invariant}. \textcolor{black}{
 The set $\QQ$  of
quadratic forms of rank 2 \rouge{is the set} of homogeneous polynomial of degree 2 that are \EA-equivalent to $x_1x_2$. The set $\QQ$ is invariant under the action of the group of affine permutations. Based on $\QQ$ and  the previous invariant $\mathfrak w$, we define an other \EA-invariant $\inv$ on $\boole(m)$, the multiset : $ \inv  (f) = \{\!\!\{ \invwalsh( f + g ) \mid g\in \QQ \}\!\!\}$}.
A \textsl{bent} function \bleu{is} a Boolean function $f$ whose Walsh transform has constant absolute value. Bent functions exist only for even $m$, and satisfy :
\begin{equation}\label{BENT}
\forall a\in \fd(m),\quad\abs({\tfr(f,a)})=\sqrt{q}=2^{m/2} \Longleftrightarrow \kappa(f)=1
\end{equation}
Two complementary notions are defined from the Walsh coefficients of a Boolean function $f$ the \textsl{linearity} $\l(f)$ and
the \textsl{non-linearity}  $\nl(f)$ with their bound :
\begin{equation*}
   \l(f) := \max_{a\in\fd(m)} \abs({\tfr(f,a)})\geq 2^{m/2}\qquad
   \nl(f):= 2^{m-1} - \frac{1}{2}\max_{a\in\fd(m)} \abs({\tfr(f,a)})\leq 2^{m-1} - 2^{m/2-1}
\end{equation*}

Bent functions have a maximal non-linearity and achieve the upper bound of non-linearity $2^{m-1}-2^{m/2-1}$. The \textsl{auto-correlation} of a Boolean function $f$ is defined for $t\in \fd(m)$ by :
\begin{equation}\label{CORR}
f\times f (t)=\sum_{x+y=t} (-1)^{f(x)+f(y)}= \frac 1 q \sum_{a\in \fd(m)} \tfr(f,a)^2 (-1)^{a.t}
\end{equation}
A vectorial $(m,n)$-function is a mapping from $\fd(m)$ into $\fd(n)$, it is 
defined by $n$ \textsl{coordinate Boolean functions}  $f_i=e_i . F(x)$
such that $F(x) = \big(f_1(x), f_2(x), \ldots, f_n(x) \big)$ with $(e_i)_{1\leq i\leq n}$ \rouge{being} canonical basis of $\fd(n)$. For any $b\in\fd(n)$,
the Boolean function $x\mapsto F_b(x) = b . F(x)$ is a \textsl{component} of $F$,
the space $\comp(F)$  of the components is generated by the coordinates of $F$. 
The degree of a vectorial function is the maximum among the degrees of its Boolean components. 
Most concepts introduced earlier for Boolean functions can be extended to vectorial functions.
We define equivalency of vectorial functions, for $F$ and $G$ two $(m,n)$-functions, $F$ and $G$ are \textsl{affine equivalent} (equivalent) if there exist an affine $(m,m)$-permutation $A$, an affine $(n,n)$-permutation $B$ such that $G(x)=(B\circ F \circ A) (x)$ ;
$F$ and $G$ are \textsl{extended affine equivalent} (\EA-equivalent) if there exist an affine $(m,m)$-permutation $A$, an affine $(n,n)$-permutation $B$ and an affine $(m,n)$-function $C$ such that $G(x)=(B\circ F \circ A) (x)+C(x)$ ;
 $F$ and $G$ are \textsl{\CCZ-equivalent} if there exists an affine permutation $\mathcal A$ on $\fd(m)\times \fd(n)$ such that $\mathcal A (\Gamma(F))=\Gamma(G)$ where $\Gamma(F)=\{ (x, F(x)) \mid x \in \fd(m)\}$ (resp. $\Gamma(G)$) is the graph of $F$ (resp. $G$).

\begin{lemma}
\label{INVARIANT}
The multiset
$\invec (F) = \{\!\!\{ \inv( f) \mid f \in \comp(F) \}\!\!\}
$ is an \EA-invariant.
\end{lemma}

The Walsh coefficient of a $(m,n)$-Function $F$ at $(a,b) \in \fd(m)\times \fd(n)$ is Walsh coefficient of its component $F_b$ :
$$\TFR (F,a,b) =\tfr (F_b,a)=\sum_{x \in \fd(m)} (-1)^{F_b(x)+a.x},$$ 
the linearity and non-linearity of a $(m,n)$-function $F$ are respectively the maximum of linearity among its components and the minimum non-linearity among its components :
$$\l(F):=\max_{a\in \fd(m), b\in\fd(n)\setminus \{0\}} \abs({\tfr(F_b,a)})\qquad \nl(F):= 2^{m-1}- \frac 1 2 \max_{a\in \fd(m), b\in\fd(n)\setminus \{0\}} \abs({\tfr(F_b,a)})$$
A $(m,n)$-function is \textsl{bent} if all its non-zero components are bent. It exists iff  $m$ is even and $n\leq m/2$. 
For $m=2k$ and $n>k$, an $(m,n)$-function $F$ is called $(m,n)$-\MNBC function see \cite{SACHAMNBC}, if it has the maximum number of bent components $2^n-2^{n-k}$. We consider the system of two equations and $r$ variables in $\fd(m)$:

\begin{equation}
                        \label{SIGMA}
x_1    + x_2    + \cdots+ x_r   =  0, \quad\text{and}\quad F(x_1) + F(x_2) + \cdots+ F(x_r) = 0.
\end{equation}      

We propose to denote by $N_r(F)$ the number of solutions,
and $T_r(F)$ the number of solutions where $x_1$, $x_2$, \ldots,
$x_r$ are not all distincts.
Let us denote
\begin{equation}\label{Q}
    Q_r(F):= \frac 1 {r!} (N_r(F) - T_r(F))
\end{equation} 

\textcolor{black}{
Using the character sum counting method, the number of solutions $N_r(F)$ 
of the above system is 
\begin{align*}
N_r( F )&=\frac{1}{2^{n+m}}\sum_{x_1, x_2, \ldots, x_r} \sum_{b \in \mathbb{F}_2^n} (-1)^{b.(F(x_1) + F(x_2) + \cdots + F(x_r))}   \sum_{a \in \mathbb{F}_2^m} (-1)^{a.(x_1 + x_2 + \cdots + x_r)} \\
&=\frac{1}{2^{n+m}} \sum_{b \in \mathbb{F}_2^n}  \sum_{a \in \mathbb{F}_2^m} \tfr(F_b,a)^r\\
\end{align*} 
In term of moments of order $r$ by 
of the components of $F$~:
\begin{equation}
\label{NUMBER}
N_r( F ) =2^{m-n} \sum_{f\in\comp(F) }  \moment(r, f ).
\end{equation}
}

When we observe the Boolean components space $\comp(F)$ of a vectorial function $F$, we are interested on the one hand in the set of their \EA-classes $C_F:=\{ \text{\EA-classes}(f) \mid  f \in \comp(F)\}$ and on the other hand in the set of all normalised 4th-order spectral moments $K_F:=\{ \kappa(f) \mid f \in \comp(F)\}$. For these two sets, we also study their cardinality and their distribution of values. Note that $\sharp K_F\leq \sharp C_F$.
We are particularly interested in vectorial functions such that $\kappa(f)$ take few  values. A vectorial function $F$ has $k$ levels of 4-th order spectral moments if the cardinality of $K_F$ is $k$.
  
\begin{example}
If $m$ is odd then all the non zero components of the power function $x^3$ in $\ffd(m)$ are \EA-equivalents, $\sharp C_F=1$, and thus $\sharp K_F=1$. 
\end{example}

\section{\APN and counting function}

Let \rouge{$F$ be an} $(m,n)$-function. For $0\ne u\in\fd(m)$, $v\in \fd(n)$, we denote $N_F(u,v)$ the number of solutions in $\fd(m)$ of the equation $F(x+u)+F(x)=v$. Note that if $x$ is a solution then $x+u$ is also a solution. Thus, $N_F(u,v)$ is even. 
\begin{equation}
 N_F(u,v)=\frac 1 {2^{m+n}} \sum_{a\in \fd(m),b\in\fd(n)} \tfr (F_b,a)^2 (-1)^{a.u} (-1)^{b.v} 
 = \frac 1 {2^n} \sum_{b\in\fd(n)} F_b\times F_b (u) (-1)^{b.v}.
\end{equation}
The \textsl{differential uniformity} of a $(m,n)$-function $F$ is  $\Delta_F:= \max\limits_{u\in\fd(m)\setminus \{0\}, v\in \fd(n)} N_F(u,v).$
A $(m,m)$-function $F$ is \textsl{almost perfect non linear} (\APN) iff it satisfies one of the following equivalent properties :
\begin{enumerate}[(i)]
\item The differential uniformity of $F$ is $\Delta_F=2$.
\item For all 2-flat $\{x,y,z,t\}\subseteq \fd(m)$,  $F(x) + F(y) + F(z) + F(t) \ne 0$.
\item  $N_4(F)  = T_4(F) = 3q^2 - 2q.$\quad (iv) $\sum\limits_{0\ne f\in \comp(F)} \kappa(f) = 2(q-1)$
\end{enumerate}

\begin{lemma} If $F$ is \APN in even dimension then \textcolor{black}{$\sharp K_F\geq 2$}.
\end{lemma}
\begin{proof}
If $f$ is non zero component of $F$ with $\sharp K_F = 1$, and (iv) implies $\kappa(f)= 2$. \textcolor{black}{By little Fermat's Theorem $q^3\equiv q \mod 3$ and $\tfr(f,a)^4\equiv \tfr(f,a)^2 \mod 3$, applying Parseval's identity, we obtain $\kappa(f)\equiv q \mod 3$, that implies $m$ odd.}
\end{proof}

\textcolor{black}{The non-existence in even dimension of \rouge{\APN functions} with a single spectral moment of order 4 naturally leads us in the next section to look for \APN-functions with two spectral moments of order 4. 
We introduce here the terminology of function with 2-spectral levels to designate a vectorial function $F$ such that $\sharp K_F=2$.}

\begin{lemma} If $F$ is \APN in dimension $m$ the number of
trivial solutions are
$$
    T_4(F) = 3q^2 - 2q,\quad  T_6(F) = q + 15 q (q-1) + 15 q (q-1) (q-2).
$$
\end{lemma}
The above Lemma can be used to give information on automorphism order.
\begin{wraptable}{l}{70mm}
\begin{tabular}{|c|ccc|cccc|}
\hline
\multirow{2}{*}{\#} & \multicolumn{3}{c|}{degree} & \multicolumn{4}{c|}{ 4th-spectral moment} \\
\cline{2-8}
&2&3&4&1&1.75&2.5&4.0\\
\hline
1  &63 &   &    & 42[1] &     &    & 21[1]\\
2  &    & 7 & 56 &    & 56[1]   &    & 7[\textcolor{green}{1}]\\
5  &1  & 62 &    & 30[2] &  & 24[2] & 9[3]\\
2  &    & 31 & 32 & 12[2] & 32[3] & 12[2] &7[\textcolor{green}{2}]\\
1  &    & 31 & 32 & 12[1] & 32[2] & 12[2] &7[\textcolor{green}{2}]\\
2  &    & 31 & 32 & 12[2] & 32[2] & 12[2] &7[\textcolor{green}{2}]\\
\hline
\end{tabular}
\end{wraptable}
\small
There are 2-spectral levels APN functions. The \CCZ-class 
of Dublin permutation is divided into 13 \EA-classes \cite{CALDERINI}, 
3 of which have 2 spectral levels, see line 1 and 2. 
This table also gives the distribution of the degrees and spectral levels of the components, specifying the number of \EA-classes. For example, the line 2 there is 2 \EA-classes which one that contains 
the Dublin permutation, 7 components are cubics and 56 are quartics. 

\small
\noindent Moreover,  56 components with  spectral moment 1.75 are in [1] \EA-class, 7 components with  spectral moment 4.0 are in [1] \EA-class and form a \textcolor{green}{vector space} of dimension 3.  
\normalsize
For an \APN vectorial function $F$, we introduce the \textsl{counting function} $n_{u}$ for a given $u\in \fd(m)$ defined for $v\in\fd(m)$
by    $$n_{u}(v)=\begin{cases}
1, & \text{if } N_F(u,v)=2;\\
0, & \text{if } N_F(u,v)=0.\\
\end{cases}
$$
This counting function is defined for each $u\in \fd(m)$ and verify for $b\in \fd(m)$:
\begin{equation}\label{LINK}
    \forall b\in\fd(m)\setminus\{0\}, \tfr (n_{u},b)= -F_b\times F_b (u)\quad \text{and}\quad \tfr (n_{u},0)=0.
\end{equation}
If all  counting functions $n_u$ are Boolean function of degree at most 1, 
the vectorial function $F$ is called \textsl{crooked}. We apply the relation \ref{LINK} 
to obtain the following observations that are mainly consequences of known classification 
of 6-bits Boolean functions.

\begin{scolie}
In dimension 6, an \APN crooked function is quadratic. 
\end{scolie}

\begin{noproof}
If $F$ is crooked, then all the $n_u$ are affine, thus all the Walsh coefficients $\tfr (n_{u},b)\equiv 0 \mod 64$, i.e. : $-F_b\times F_b (u) \equiv 0 \mod 64$. One can check from the classification of Boolean 6-bit functions that all function 
with such a correlation must be quadratics.
\end{noproof}
\begin{scolie}
All the counting functions of 6-bit \APN function of degree 6 are quintic. 
\end{scolie}
\begin{noproof}
On one hand, all the correlation coefficients of a Boolean 
function of degree 6 are congruent to 4 modulo 8. On another
hand, all the Walsh coefficients  of balanced function
of degree 4 are to zero modulo 8.
\end{noproof}
\begin{scolie}
In dimension 6, if $F$ is a \MNBC function then it is not \APN. 
\end{scolie}
\begin{noproof}
One can deduce that from \cite{SACHAMNBC}. We propose an alternative proof
from the classification of 6-bit functions. Let $F$ be a putative \MNBC and \APN in 6 variables. 
Any of its counting  function has at 7 non zero
Walsh coefficients. A simple observation of 150356 
class of Boolean function in $\boole(6)$ shows
theses  counting functions are affine or almostlinear
with Walsh coefficient takes values in
$\{0, 32, 64\}$. The auto-correlation 
values of components  share the same set, 
and again, a direct observation shows they
must have degree at most 3.
\end{noproof}
\section{Function with 2-spectral levels}

We observe the existence  of two spectral level function in each 
of the 14 known \APN \CCZ-classes in dimension 6, and we decided to search
for other examples by extension process. A vectorial \APN function $F$ is with \textsl{2-spectral levels} if the normalized 4th-order spectral moments of its components take 2 distinct values $\alpha$ and $\beta$. In this case, \begin{equation}\label{AB}
        \alpha A + \beta B = 2(q-1), \quad A + B = q-1 ;
\end{equation} where $A$ (resp. $B$) is the number of components $f$ of $F$ such that $\kappa(f)=\alpha$ (resp. $\kappa(f)=\beta$). We suppose that $\alpha < \beta$ and we say $F$ is a function of type $(\alpha,\beta)$.
Using the classification of Boolean functions,  among 293 values of $\kappa$, we found 62 possible pairs satisfying \ref{AB}~, involving function of degree less or equal to 4 :
$$
\hbox{
\tiny
\begin{tabular}{|r|c|c|c|c|r|c|c|c|c|}
\hline
$\alpha$  &A&$\deg$&$\sharp$& $\beta$ &B&$\deg$&$\sharp$\\
\hline
1.0&42 &23. &{4} &4.0&21 &234 &{86}\\
1.0&60 &23. &{4} &22.0& 3 &.3. &{1}\\
1.0&56 &23. &{4} &10.0& 7 &.3. &{1}\\
1.0&49 &23. &{4} &5.50&14 &.34 &{29}\\
1.0&21 &23. &{4} &2.50&42 &.34 &{216}\\
1.0&57 &23. &{4} &11.50& 6 &.34 &{5}\\
1.0&35 &23. &{4} &3.250&28 &.34 &{191}\\
1.0&15 &23. &{4} &2.3125 &48 &.34 &{214}\\
1.0&51 &23. &{4} &6.250&12 &..4 &{13}\\
1.0&47 &23. &{4} &4.9375 &16 &..4 &{37}\\
1.0&39 &23. &{4} &3.6250&24 &..4 &{67}\\
1.0& 7 &23. &{4} &2.1250&56 &..4 &{49}\\
1.0&55 &23. &{4} &8.8750& 8 &..4 &{2}\\
\hline
\end{tabular}
\begin{tabular}{|r|c|c|c|c|r|c|c|c|c|}
\hline
$\alpha$  &A&$\deg$&$\sharp$& $\beta$ &B&$\deg$&$\sharp$\\
\hline
\rowcolor[gray]{.8} 1.750&56 &..4 &{8} &4.0& 7 &234 &{86}\\
1.750&42 &..4 &{8} &2.50&21 &.34 &{216}\\
1.750&60 &..4 &{8} &7.0& 3 &.34 &{3}\\
1.750&51 &..4 &{8} &3.0625 &12 &.34 &{321}\\
1.750&35 &..4 &{8} &2.3125 &28 &.34 &{214}\\
1.750&59 &..4 &{8} &5.6875 & 4 &..4 &{25}\\
1.750&21 &..4 &{8} &2.1250&42 &..4 &{49}\\
1.750&49 &..4 &{8} &2.8750&14 &..4 &{119}\\
1.750&57 &..4 &{8} &4.3750& 6 &..4 &{34}\\
1.9375 &56 &..4 &{54} &2.50& 7 &.34 &{216}\\
1.9375 &60 &..4 &{54} &3.250& 3 &.34 &{191}\\
1.9375 &42 &..4 &{54} &2.1250&21 &..4 &{49}\\
1.9375 &62 &..4 &{54} &5.8750& 1 &..4 &{19}\\
\hline
\end{tabular}
}
$$
If we restrict our attention to the case where the set of components such that $\kappa(f)=\alpha$ or $\kappa(f)=\beta$ forms a vector space, thus $A$ or $B$ is a power of 2 minus 1, 
we obtain 6 possible pairs listed in the Table \ref{PAIR}.
The Table describes the structure of a potential vectorial function of type $(\alpha,\beta)$.
For example, the 4-th line corresponds  to the pair $(1.75,4)$, for which we have $A=56$ and $B=7$.
The components corresponding to $\alpha=1.75$ (resp. $\beta=4.0)$ must be chosen from 8 (resp. 86) classes of Boolean functions of degree 4 (resp. 2, 3 and 4). We remark that the permutation obtained in \cite{DILLON} is of type (1.75,4) and corresponds to this line.
\begin{table}[htbp]
\small
\caption{\label{PAIR} Six possible pairs.}
\begin{tabular}{|l|c|c|c|l|c|c|c|c|}
\hline
$\alpha$ &$A$ &degree  &$\sharp$classes &$\beta$  &$B$  &degree &$\sharp$classes\\ 
\hline
\hline
\rowcolor{red} 1.0000 &(56) &23... &{4} & 10.0000 &( 7) &.3... &{1}\\
\rowcolor{red} 1.0000 &(15) &23... &{4} & 2.3125 &(48) &.34.. &{214}\\
 1.0000 &( 7) &23... &{4} & 2.1250 &(56) &..4.. &{49}\\
\rowcolor[gray]{.8}  1.7500 &(56) &..4.. &{8} & 4.0000 &( 7) &234.. &{86}\\
 1.9375 &(56) &..4.. &{54} & 2.5000 &( 7) &.34.. &{216}\\
 1.9375 &(62) &..4.. &{54} & 5.8750 &( 1) &..4.. &{19}\\
 \hline
\end{tabular}
\end{table}
The pair $(1,10)$  corresponding to the first line of the table,
describes an \APN and \MNBC vectorial function of degree less than 3. It follows from
the \cite{SACHAMNBC} it does not exists. 
The second line describes a vector function with a bent-space of dimension 4, that is impossible.
Our objective is \rouge{to find} new vectorial functions of type $(\alpha,\beta)$ in the Table \ref{PAIR}. 
\textcolor{black}{The grey line} covers the case of the Dublin permutation,
and potentially new \CCZ-\APN classes because of the large number 
of possibility in term of classes. \textcolor{black}{
In light of the numerical searches carried out during the last decades, 
the existence of \APN  functions of degree $\geq 5$ seems unlikely.}
In the next section, we decided to
restrict the area of exploration in the space of vectorial quartic functions of type $(1.75,4.0)$, also assuming degree
4 for the $1.75$-components and degree $\leq 3$ for the $4.0$-components.

\section{Numerical investigation}

An extension $G$ of $F$ is obtained by adding some coordinate functions, in
that case $\comp(F)$ becomes a subspace of components space of $G$.

\begin{lemma}\label{DELTA}
If a $(m,n)$-function $F$ has an \APN extension then $\Delta_F \leq 2^{m-n+1}$.
\end{lemma}

The vectorial $(m,m-2)$-function $F$ has an \APN extension , if and only if,
for all $(x,y,z,t)\in Q_F$ the system \rouge{of} quadratic equations :
\begin{equation}
\label{SYSTEM}
               g(x) + g(y) + g(z) + g(t) \not= 0 \Leftrightarrow  \big(g(x) + g(y) + g(z) + g(t)\big)^3 = 1.
\end{equation}
is solvable in $\fd(4)$. We remark that   $\big(g(x) + g(y) + g(z) + g(t)\big)^3 $ equal to :
\begin{equation*}
          x^3 + y^3 + z^3 + t^3 + xy(x+y) + xz(x+z)+ xt(x+t) + yz(y+z) + yt(y+t)+ zt(z+t).
\end{equation*}
so we can transform system (\ref{SYSTEM}) in an affine system $S_F$ with $N$ equations and $q(q+1)/2$ unknowns,
 introducing  $q$ Boolean variables $x^3$, and   $q(q-1)/2$ variables $xy(x+y)$.

\begin{lemma}\label{XYZT}
        If the affine system $S_F$ has no solution then $F$ has no \APN extension.
\end{lemma}

We say that an $(m,m-2)$-vectorial function passes the extension test if it satisfies conditions of Lemma \ref{DELTA} and Lemma \ref{XYZT}.  Even it is an hard task, it is possible to use 
the following procedure to "classify" all \APN functions 
of type $(\alpha,\beta)$ that are quartic extensions of 
a $(6,3)$-vectorial cubic. Let $\mathcal E$ be a set of  $(m,n)$-functions. We define
$\ext( \mathcal E )$ as the set of extensions $(F,f)$
having $(\alpha,\beta)$ type that satisfy Lemma \ref{DELTA} and 
"filtered" by invariant $\invec$.  Starting from $\mathcal E_0 :=\{ h\}$
where $\deg(h)=4$ and $\kappa(f)=\alpha$, 
we contruct $\mathcal E_1 = \ext( \mathcal E_0 )$,  $\mathcal E_2 = \ext( \mathcal E_1 )$,
and  $\mathcal E_3 = \ext( \mathcal E_2 )$. We keep the $(6,4)$-function passing
the extension test, and we terminate by a backtracking algorithm to identify
\APN extension, and then 2-level \APN functions.

Applying the procedure using the invariant of Lemma \ref{INVARIANT} 
for the pair (1.75,4), \textcolor{black}{there are 8 quartic classes to initialise the construction process, 4 of which produce \APN functions.} \bleu{The 506880 \APN functions obtained 
after the backtracking phase are not necessarily at 2-spectral levels, but all 16384 functions of type (1.75,4) are ultimately \CCZ-equivalent 
to Dublin permutation.}

\textcolor{black}{
\section{Conclusion}
We introduced the new concept of spectral level to guide our research towards the construction of \APN functions that are strongly structured from a spectral perspective, following the example of the Dublin permutation.
We described a procedure for exploring 2-spectral level \APN functions in dimension 6, and applied it to the search for functions of type (1.75, 4.0), establishing the uniqueness of the \CCZ-class of the Dublin permutation.
This procedure relies on an original extension test, which also validated the backtracking approach. It may be applied to higher dimensions.
The invariant used to limit the combinatorial explosion requires further refinement to allow a complete exploration of all possible pairs of functions at 2-spectral levels.}
\normalsize
\nocite{*}
\bibliographystyle{plain} 
\bibliography{bfa-2025-abstract.bib}
\end{document}